\documentclass[10pt]{amsart}
\usepackage[utf8]{inputenc}
\usepackage{amsmath,amssymb,amsthm}
\usepackage{enumitem}
\usepackage[hidelinks]{hyperref}
\usepackage{cleveref}
\usepackage{tikz}
\usepackage{enumitem}
\usepackage{comment}
\usetikzlibrary{patterns}
\DeclareMathOperator{\var}{\mathrm{Var}}
\DeclareMathOperator{\tr}{tr}

\DeclareMathOperator*{\esssup}{ess\,sup}
\DeclareMathOperator*{\essinf}{ess\,inf}
\newcommand{\Z}{\mathbb{Z}}
\newcommand{\R}{\mathbb{R}}

\renewcommand{\P}{\mathbb{P}}
\newcommand{\E}{\mathbb{E}}

\newcommand{\ve}{\varepsilon}

\newtheorem{thm}{Theorem}[section]

\newtheorem{claim}[thm]{Claim}
\newtheorem{lem}[thm]{Lemma}
\newtheorem{cor}[thm]{Corollary}
\newtheorem{prop}[thm]{Proposition}

\theoremstyle{definition}

\newtheorem{deffo}[thm]{Definition}
\newtheorem{remm}[thm]{Remark}
\newtheorem{ex}[thm]{Example}
\numberwithin{equation}{section}

\tikzstyle{dot}=[fill=black, draw=black, shape=circle]

\title[Non-stationary localization in three dimensions]{Localization for non-stationary Anderson models in three dimensions}

\author{Omar Hurtado}

\begin{document}
	
\begin{abstract}
    We prove localization (near the bottom of the spectrum) for certain non-stationary variants of the Anderson model in three dimensions. More specifically, we prove a Wegner estimate, which implies localization by existing work. Two key inputs are a deterministic quantitative unique continuation theorem  by Li and Zhang \cite{li2022anderson} and some combinatorial decompositions/bounds for non-stationary random potentials proved in \cite{hurtado2026localization}.
\end{abstract}
\maketitle

\section{Introduction}
\subsection{Main results}
Throughout, we are concerned with non-stationary variants of the Anderson model on the three-dimensional lattice, i.e. operators on $\ell^2(\Z^3)$ of the form
\begin{equation}\label{nsham}
    H = -\Delta + V
\end{equation}
where $\Delta$ is the discrete Laplacian (i.e. the graph Laplacian for the three-dimensional cubic lattice) and $V$ is a random potential acting by multiplication, with $V_n$ independent but not necessarily identically distributed.

Our main result is localization at the bottom of the spectrum:
\begin{thm}\label{mainlocalthm}
    Let $(V_n)_{n\in\Z^3}$ be an independent potential satisfying
    \begin{itemize}
        \item $\P[0 \leq V_n \leq M] = 1$ for some $M > 0$ and all $n \in \Z^3$
        \item $\sigma^2:=\inf_{n\in\Z^3} \var V_n > 0$
    \end{itemize}
    Then there exists $E_0 = E_0(M,\sigma^2)  > 0$ such that $H$ defined by \eqref{nsham} is almost surely Anderson localized in the interval $[0,E_0]$; that is, there is no continuous spectrum in $[0,E_0]$ and all the eigenfunctions corresponding to energies in this range are exponentially decaying.
\end{thm}

In fact, a stronger result, known as strong dynamical localization in expectation, will also follow from our proof. Note that $H$ is self-adjoint and so in particular expressions such as $\chi_{[0,E_0]}(H)$ and $e^{-itH}$ are well-defined via functional calculus, and the latter defines the time evolution of the physical system associated to the Hamiltonian $H$. If we define a position operator by $[\langle X \rangle \psi](n)= |n| \psi(n)$, then strong dynamical localization gives strong control on the moments of the position, giving quantitative control over transport.
\begin{thm}\label{sdl}
Let $(V_n)_{n\in\Z^3}$ be a random potential satisfying the hypotheses in \Cref{mainlocalthm}. Then there exists $E_0  = E_0(M,\sigma^2)> 0$ such that for any $b > 0$ and $s> 0$ sufficiently small (smallness depending on $b$ and the potential), we have the following moment bound:
\begin{align*}
    \E\left[ \sup_{t \geq 0}\|\langle X\rangle^b e^{-itH} \chi_{[0,E_0]}(H)\delta_0\|^s \right] < \infty
\end{align*}
    
\end{thm}

Both of these results, by work of Rangamani and Zhu \cite{rangamani2025dynamical}, follow from certain bounds on the finite volume resolvents of $H$. Specifically, we let $\Lambda$ be a ``cube'' in $\Z^3$ of radius $L$. If we let $H_\Lambda$ denote the restriction of $H$ to $\Lambda$ (in a sense made precise in \Cref{prelim}) the required resolvent bounds are the following.

\begin{thm}\label{mainestthm}
    Given any independent random potential $(V_n)_{n\in\Z^3}$ satisfying the hypotheses of \Cref{mainlocalthm}, there are $E_0, \kappa_\ast, \ve_\ast$ and $m_\ast$ positive (depending on $M$, $\sigma^2$) such that for all sufficiently large $L$, the corresponding operator $H$ satisfies
    \begin{equation*}
        \P[|(H_\Lambda-E)^{-1}(x,y)| \leq \exp(L^{1-\ve_\ast} - m_\ast|x-y|)] \geq 1- L^{-\kappa_\ast}
    \end{equation*}
    for any $L$-cube $\Lambda$, $x,y \in \Lambda$ and energy $E \in [0,E_0]$.
\end{thm}

The kind of resolvent bound appearing in \Cref{mainestthm} has been obtained in the setting of singular noise via a variant of the famous multiscale analysis (MSA) introduced by Bourgain and Kenig in \cite{bourgain2005localization}. (The original MSA was introduced by Fr\"ohlich and Spencer in \cite{frohlich1983absence}, and had been successively iterated on in many works, only some of which are \cite{carmona1987anderson, von1989new,germinet2001bootstrap}.) 

While we don't formulate precisely our most important technical result in the introduction (\Cref{bigweg}), it essentially bounds the probability that an eigenvalue of a truncation of the operator lies in a small interval under the assumption that the eigenfunctions are somewhat localized. From this and a more or less routine initial scale estimate, one can carry out a multiscale analysis (MSA) argument. We do not explicitly carry out this analysis; the essential argument goes back to \cite{bourgain2005localization}, and the changes necessary for the discrete case were made in \cite{ding2020localization}. Non-stationarity does not meaningfully affect the framework, which also appears in appendices in \cite{hurtado2026localization,li2022anderson,li2022anderson2}. The main place in which stationarity is vital in previous arguments is in the use of combinatorial arguments to obtain a ``Wegner estimate'' of sorts, the bound on the probability of finding an eigenvalue in a small interval alluded to earlier. Following \cite{hurtado2026localization}, we use Bernoulli decompositions satisfying certain uniform estimates to essentially reduce the general non-stationary case to the non-stationary Bernoulli case, which is amenable to combinatorial methods.

Before proceeding further, some remarks are in order:

\begin{remm}
    \,
    \begin{enumerate}
        \item In \Cref{mainlocalthm}, it is not so important that the potential $V$ is non-negative so much as that it is bounded above and below; by an additive normalization any such potential can be made non-negative. Such a normalization does not affect the spectral theory of $H$ in any meaningful way. The non-negativity assumption however does allow us to more succinctly state results.
        \item It is not an immediate consequence of the assumptions in \Cref{mainlocalthm} that $[0,E_0]$ is contained in the spectrum, i.e. that our theorem is non-trivial. However, our result is non-trivial in many specific cases of interest. This is discussed in \Cref{ntss}.
    \end{enumerate}

\end{remm}

\subsection{Background}
The Anderson model is one of the most heavily studied mathematical models among those arising from condensed matter physics, and it is impossible to summarize the history in a brief manner; we recommend \cite{aizenman2015random,cycon2009schrodinger} for broad treatments and sketch out some of the work which is most important for understanding the place of the current work in the broader picture. We will still try and summarize briefly some of the more relevant works in the mathematical study of Anderson localization.

While the non-stationarity is the main novelty of the current paper, a salient property of the class of potentials under consideration is that they may be very singular, even to the point of admitting point masses. Indeed, we point out that approaches to proving localization in higher dimensions (in the expected energy-coupling regimes) were developed already in the mid eighties; the celebrated multiscale analysis (MSA) of Fr\"ohlich and Spencer \cite{frohlich1983absence} was used by those same authors together with Martinelli and Scoppola to prove localization in higher dimensions \cite{frohlich1985constructive}. However, this required a regularity assumption on the potential, in particular excluding Bernoulli potentials or any potentials with atoms. In one dimension, MSA was combined with transfer matrix methods by Carmona, Klein and Martinelli to prove localization for singular potentials in dimension one \cite{carmona1987anderson}. (Various approaches eschewing MSA have since appeared \cite{Jitomirskaya2019,bucaj2019localization,gorodetski2021parametric,shubin1998some}.) The multi-dimensional case remained completely open until a breakthrough of Bourgain and Kenig \cite{bourgain2005localization} showed localization at the bottom of the spectrum for a continuum analogue of the Bernoulli-Anderson model, exploiting quantitative unique continuation bounds for eigenfunctions and combinatorial bounds. 

Technical obstructions prevented a proof of localization for discrete models, until the work of Ding and Smart \cite{ding2020localization} and Li and Zhang \cite{li2022anderson} proved the same for discrete models in two and three dimensions respectively, combining the Bourgain-Kenig approach with new discrete continuation results. Localization for singular potentials in dimensions four and higher still remains open in the discrete setting.

Indeed, without singularity, such results for non-stationary potentials are much easier to prove. In principle, it seems reasonable that via the versions of MSA available before the Bourgain-Kenig breakthrough one could have proved a version of our results which replaces the variance lower bound with uniform H\"older continuity, i.e. the strictly weaker condition that there exist $C,\alpha > 0$ such that
\begin{equation}\label{holderreg}
    \mathbb{P}[|V_n - x| \leq \ve] \leq C\ve^\alpha
\end{equation}
for all $x \in [0,M], \ve>0, n \in \mathbb{Z}^d$. However, we are not aware of any such result in the literature. Under the even stronger assumption that the site potentials $V_n$ have uniformly bounded density functions with respect to Lebesgue measure, i.e. the case
$\alpha = 1$ in \eqref{holderreg}, the result is a straightforward consequence of work of Aizenman and Molchanov \cite{aizenman1993localization}, specifically their celebrated fractional moment method.
\begin{remm}
    In comparing the variance lower bound hypothesis with stronger anti-concentration hypotheses like H\"older continuity and absolute continuity, it is worth pointing out that the hypotheses that $0 \leq V_n \leq M$ with probability one and that $\var V_n \geq \sigma^2 > 0$ together imply a ``single scale'' anti-concentration; specifically for any open interval $I$ of radius $\sigma/2$, we have
    \begin{equation*}
        \P[V_n \in I] \leq 1- \frac{9}{16}\frac{\sigma^4}{\sigma^4 + M^4}
    \end{equation*}
    This is a straightforward application of the second moment method, see \cite[Proposition 2.5]{hurtado2026localization}. While this fact will not play any role explicitly in the current work, it is exploited crucially in the proof of \Cref{berdecomp}.
\end{remm}

The first results on non-stationary potentials of the kind considered here were one-dimensional results. In \cite{hurtado2023lifting}, a localization result was proved for non-stationary potentials which were in some sense ``almost'' stationary; a much more general result was proved by Gorodetski and Kleptsyn, essentially proving the analogue of \Cref{mainlocalthm} and what is called ``semi-uniform dynamical localization'' in one dimension in \cite{gorodetski2025non}. The latter result has been generalized to certain unbounded non-stationary potentials by Zieber in \cite{zieber2026localization}. All these works rely in a fundamental way on the transfer matrix method, a fundamentally one-dimensional approach; indeed, one of the main inputs in the latter two works is the non-stationary Furstenberg theory developed in \cite{gorodetski2022non}. These results all prove localization throughout the spectrum, whereas our result only proves it at the bottom. (For ``weak'' randomness, localization throughout the spectrum is not expected in three dimensions.)

In higher dimensions, the only result for non-stationary singular potentials of this type is \cite{hurtado2026localization}, where the analogue of \Cref{mainlocalthm}, \Cref{mainestthm} and \Cref{sdl} was shown in two dimensions. (Unlike the three-dimensional picture, localization throughout the spectrum is expected in two dimensions even for ``weak'' randomness; however this has not been proved even for much more constrained potentials.) This result used the unique continuation framework introduced by \cite{bourgain2005localization} and used in \cite{klein2012comprehensive,ding2020localization,li2022anderson}, adapting much of \cite{ding2020localization} specifically to the non-stationary setting, including a proof of unique continuation for eigenvectors of Schr\"odinger operators on $\Z^2$ with non-stationary potentials. The present work implements the approach there, combining it with the unique continuation bounds of Li and Zhang. Said bounds are deterministic, and so in particular are robust to changes to the structure of the noise, unlike the two-dimensional case.

\section*{Acknowledgements} The author was supported in part by NSF grant DMS-2503569.

\section{Preliminaries}\label{prelim}

\subsection{Notation}
We use $\Delta$ to denote the discrete Laplacian in whatever context we currently discuss. In particular, there is a natural graph structure on $\Z^d$ where there is an edge between $n = (n_1,\dots,n_d)$ and $n' = (n_1',\dots,n_d')$ (henceforth denoted $n \sim n'$) if and only if $\|n-n'\|_{\ell^1(\Z^d)} = \sum_{k=1}^d |n_i - n_i'| = 1$; this gives the familiar cubic lattice in $\Z^d$. Then $[\Delta \psi](n)  = 2d\psi(n) - \sum_{n' \sim n} \psi(n')$. Random potentials always act by multiplication, i.e. $[V\psi](n) = V_n\psi(n)$.

We use $\P$ and $\E$ to denote probabilities of events and expectations of random variables respectively. Given a $\sigma$-algebra $\mathcal{F}$ or a random variable $Y$ or an event $A$, we let expressions like $\P[A|\mathcal{F}]$, $\E[X|A]$ and $\E[X|Y]$ denote conditional probabilities and expectations, with the last of these meaning $\mathbb{E}[X|\sigma(Y)]$.

Throughout, we will also use $C$ and $c$ to denote large and small constants respectively. Usually, these will depend on some parameters of the particular potential under consideration, and we will try to indicate this dependence with subscripts, e.g. $C_\ve$ is a large constant which depends on a parameter $\ve$. We will occasionally suppress some of these dependencies for legibility; we do this explicitly in the statement of our main technical result \Cref{bigweg}.

Throughout, we will also be considering truncations of our operator to finite volume cubes $\Lambda \subset \mathbb{Z}^3$. By ``cubes'' we mean sets of the form $\{m \in \Z^3\,: |m-n|\leq L\,\}$ for some integer $L > 0$ and $n \in \Z^3$. In particular, we use $\Lambda_L(n)$ to denote precisely the cube specified above, and say it has side length $2L+1$. For some integer $L$, we will call a cube an $L$-cube if it is a cube of the form $\Lambda_L(n)$ for some $L \in \mathbb{N}$ and $n \in \Z^3$. Associated to any cube $\Lambda$ (or more generally any subset of $\Z^3$) there is a natural orthogonal projection $P_\Lambda$ acting on $\ell^2(\Z^3)$ by projection onto the subspace $\ell^2(\Lambda)$. Explicitly:
\begin{align*}
    [P_\Lambda\psi](n) = \begin{cases}
        \psi(n) &\text{ if } n \in \Lambda\\
        0 &\text{ otherwise}
    \end{cases}
\end{align*}
For $\Lambda$ a cube of side length $2L+1$, $H_\Lambda := P_\Lambda H P_\Lambda$ is a finite rank operator on $\ell^2(\Z^3)$ which can be naturally identified with an operator acting on the $(2L+1)^3$-dimensional space $\ell^2(\Lambda)$. It is worth stressing that it is finite volume resolvents which we are interested in bounding; by expressions like $(H_\Lambda - E)^{-1}$ we always mean the $\ell^2(\Lambda)$ inverse of $H_\Lambda - E$ understood as an operator on $\ell^2(\Lambda)$. We are never interested in ``infinite volume'' resolvents corresponding to finite volume truncations.

\subsection{Non-triviality and the bottom of the spectrum}\label{ntss}
As was discussed in the introduction, our main theorem can be trivial for certain potentials satisfying its hypotheses; it guarantees an interval $[0,E_0]$ in which we see localization, but cannot guarantee the spectrum actually intersects this interval. And there are at least contrived examples where the spectrum really does not intersect this interval; if $V_n \geq E_0+1$ almost surely for all $n$, then clearly the spectrum is contained in $[E_0+1, \infty)$ almost surely. (The $E_0$ we obtain in fact depends on the parameters of the potential, but the one obtained by our arguments will actually decrease as the essential supremum potential increases.) However, if $\essinf V_n = 0$ (or is less than $E_0$) for all $n$ in a uniform way, this cannot happen. (In fact, one can allow a natural density zero subset of sites where things go wrong and still obtain a non-triviality result.)
\begin{prop}
    Fix an independent potential $(V_n)_{n\in\Z^3}$. Let $x \in \R$ and $\ve > 0$ be such that 
    \begin{equation}\label{lowessinf}\mathbb{P}[V_n \leq x] \geq \ve\end{equation} for all $n$ in a natural density 1 subset of $\Z^3$. More precisely, if $A$ is the set of $n$ satisfying \eqref{lowessinf}, then we presume \begin{equation*}\lim_{k\rightarrow \infty} \frac{|A \cap \{-k,-k+1,\dots,k-1,k\}^3|}{(2k+1)^3} = 1\end{equation*} Then $\sigma(H)$ almost surely contains $[x,x+4]$. In particular, if $x < E_0$, then the interval $[x,\min\{x+4,E_0\}]$ is contained in the spectrum.
\end{prop}
This is a three-dimensional analogue of \cite[Proposition 2.8]{hurtado2026localization} and the proof presented there applies to the current situation with minimal changes. Note that the energy interval $[0,E_0]$ on which our method proves localization can be chosen to depend on the potential in a coarse way, i.e. the dependence is entirely via the parameters $M$ (the uniform upper bound) and $\sigma^2>0$ (the variance lower bound). In order to keep our examples succinct, we will only consider the case where we can take $x=0$; the $E_0$ obtained by our methods is usually quite small, so beyond showing that our results are not made trivial by arbitrarily small perturbations, we are not sure there is much gain in being able to consider potentials where e.g. sometimes $\essinf V_n = 10^{-1000}$ instead of $0$.

\begin{ex}
    If $(V_n)_{n\in\Z^3}$ is i.i.d. and bounded but not almost surely constant, then the associated operator exhibits localization at the bottom of the spectrum. However, this already follows via the methods of \cite{li2022anderson}.
\end{ex}

\begin{ex}
    If $(V_n)_{n\in\Z^3}$ are independent Bernoulli variables, with $V_n \sim \mathrm{Ber}(p_n)$, then as long as $\inf_{n\in\Z^3} \min\{p_n, 1-p_n\}>0$, the associated operator is localized at the bottom of the spectrum.
\end{ex}

\begin{ex}
    We let $\mu_1,\mu_2,\dots, \mu_N$ be a finite family of distributions with bounded support contained in $[0,\infty)$ such that $0 \in \mathrm{supp }\,\mu_k$ for all $k$. Then any independent potential such that for all $n$ the law of $V_n$ is among the collection $\{\mu_1,\dots,\mu_N\}$ is localized at the bottom of the spectrum.
\end{ex}

Two particularly interesting and natural examples of the latter are periodic arrangements of noise and interfaces. More precisely, for an example of the former, let $n = (n_1,n_2,n_3)\in \Z^3$, and consider a potential such that $V_n \sim \mathrm{Ber}(1/2)$ if $n_1+n_2+n_3$ is even and $V_n \sim \mathrm{Unif}(0,1)$ if $n_1+n_2+n_3$ is odd. For the latter consider e.g. $V_n \sim \mathrm{Ber}(1/2)$ if $n_1 < 0$ and $V_n \sim \mathrm{Unif}(0,1)$ if $n_1 \geq 0$. Understanding these examples even in the case where the essential infima are not the same (or very close to the same) seems an interesting question, but our methods currently cannot say anything about e.g. the case where $V_n = \mathrm{Ber}(1/2)$ for $n_1+n_2+n_3$ odd and $V_n \sim \mathrm{Ber}(1/2) + 1$ for $n_1+n_2+n_3$ even.

\subsection{Bernoulli Decompositions and Combinatorics}

No new results appear in this subsection; we recall various key probabilistic and combinatorial results which are central to the proof. A key step in the proof requires decomposing variables as integrals over (inhomogeneous) ensembles of Bernoulli variables in order to leverage certain combinatorial estimates. We recall here the necessary notions and facts.
\begin{deffo}
    Given a random variable $X$, we say a pair of measurable functions $Y,Z: (0,1) \to \R$ is a $p$-Bernoulli decomposition of $X$ if $X \overset{\mathcal{D}}{=} Y(t) + \xi Z(t)$, where $\overset{\mathcal{D}}{=}$ denotes equality in distribution, $t$ is taken to be a variable uniformly distributed on $(0,1)$ and $\xi$ is a Bernoulli variable with $\P[\xi =1] = p$, $\P[\xi = 0] = 1-p$, and $t$ and $\xi$ are independent.
\end{deffo}
In order to leverage bounds effectively, we require that our decompositions be uniform in a certain sense;
we make use of the following fact \cite[Theorem 4.5]{hurtado2026localization}:
\begin{thm}[\cite{hurtado2026localization}]\label{berdecomp}
    For any positive $M$ and $\sigma^2$, there are $p_-(M,\sigma^2) > 0$, $p_+(M,\sigma^2) < 1$ and $\iota(M,\sigma^2) > 0$ such that any random variable $V$ with $0 \leq V \leq M$ almost surely and $\var V \geq \sigma^2$ admits a $p$-Bernoulli decomposition with
    \begin{enumerate}
        \item $p_- \leq p \leq p_+$
        \item $Z(t) \geq \iota > 0$ for all $t \in (0,1)$
    \end{enumerate}

    Moreover, for $M\geq 1$ and $\sigma/M$ sufficiently small, e.g. $\sigma^2 \leq 10^{-8}M^2$, one can take $p_-,p_+, \iota$ to satisfy the following bounds: 
    \begin{enumerate}
        \item $\min\{p_-,1-p_+\} \geq \frac{\sigma^5}{2M^4}$
        \item $\iota \geq \frac{\sigma^{10}}{4M^9}$
    \end{enumerate}
\end{thm}
We also make use of a generalization of the Sperner property from combinatorial set theory, what we call $\kappa$-Sperner for some parameter $\kappa \in (0,1]$. To our knowledge, this notion first appeared in \cite{ding2020localization}.
\begin{deffo}
    We call a family $\mathcal{A} \subset 2^{\{1,\dots,N\}}$ of subsets of $\{1,\dots,N\}$ a $\kappa$-Sperner family if for every $A \in \mathcal{A}$ there is some $B(A) \subset A^C$ such that $|B(A)| \geq \kappa |A^C|$ and so that any $A' \neq A$ also in $\mathcal{A}$, we have $A' \cap B(A) = \varnothing$.
\end{deffo}
Taking $\kappa = 1$, we get $B(A) = A^C$ and recover the standard notion of a Sperner set, i.e. a family which is an anti-chain with respect to the inclusion ordering on the power set $2^{\{1,\dots,N\}}$. Given an ensemble of $N$ Bernoulli variables, one can identify realizations of them with subsets of $\{1,\dots,N\}$ in a natural way; by identifying the random $\{0,1\}$ valued vector $\boldsymbol{\xi} = (\xi_1,\dots,\xi_N)$ with the subset of $A$ defined by $n \in A$ if and only if $\xi_n = 1$, any ensemble of $N$ (ordered) Bernoulli variables gives rise to a random distribution on $2^{\{1,\dots,N\}}$. Abusing notation, we let $\boldsymbol{\xi}$ denote the associated subset of $\{1,\dots,N\}$.
We need the following bound \cite[Lemma 4.9]{hurtado2026localization}:
\begin{lem}[\cite{hurtado2026localization}]\label{kappaspern}
    Let $\xi_1,\dots,\xi_N$ be jointly independent Bernoulli variables, with $0 < \beta \leq \min\{\P[\xi_n = 0],\P[\xi_n = 1]\}$ for some $\beta \in (0,1/2]$ and all $n= 1,\dots,N$. Then there is some universal constant $C$ such that for any $\kappa$-Sperner family $\mathcal{A}$, we have:
    \begin{equation}
        \P[\boldsymbol{\xi} \in \mathcal{A}] \leq \frac{C}{\beta^{5/2}\kappa N^{1/2}}
    \end{equation}
\end{lem}
This bound generalizes (and is inspired by) \cite[Theorem 4.2]{ding2020localization} and some results from \cite{aizenman2009bernoulli} and makes crucial use of ideas and results from \cite{yehuda2026lym}; this is discussed in more detail in \cite{hurtado2026localization}.

\subsection{Geometric notions and unique continuation}

In this section, we present the necessary unique continuation result from \cite{li2022anderson}. The statement requires the notions of graded sets and normal subsets. These are geometric constraints on sets of points; they are not important for our argument, so we will state the unique continuation bound without explicitly defining these. The definitions appear in \cite[Definition 3.3]{li2022anderson}.

\begin{thm}[\cite{li2022anderson}]\label{ucfinal}
    There is some constant $p > 3/2$ such that for any positive integer $N$, any positive real $K$, and sufficiently small positive real $\ve$, there are constants $C_{\ve,K}$ and $C_{\ve, N}$ making the following hold:

    Let $L > C_{\ve,N}^4$ and $u, V$ be functions $\Z^3 \rightarrow \R$ satisfying
    \begin{equation*}
        \Delta u = Vu
    \end{equation*}
    with $\|V\|_{\ell^\infty(\Z^3)} \leq K$. Let $\vec{\ell}$ be a vector of positive reals and $F \subset \R^3$ be $(N,\vec{\ell},\frac{1}{\ve},\ve)$-graded with first length scale $\ell_1 > C_{\ve,N}$. Assume also that $F$ is $(1,\ve)$-normal in $\Lambda_L(\mathbf{0})$.

    Then
    \begin{equation}\label{ucestimate}
        \left|\left\{n \in \Lambda_L(\mathbf{0})\setminus F\,:\,u(n) \geq e^{-C_{\ve,K}L}|u(\mathbf{0})|\right \}\right| \geq L^p
    \end{equation}
\end{thm}

\begin{remm}
    Since $\Delta$ is translation invariant and of course translation preserves the $\ell^\infty(\Z^3)$ norm of $V$, this result holds on boxes centered at any point. Similarly, throughout this work, we will for simplicity often assume a box is centered at the origin; by translation invariance of our (deterministic) hypotheses this causes no loss of generality.
\end{remm}
Another important ingredient from the work of Li and Zhang is a certain cone lemma \cite[Lemma 2.4]{li2022anderson}, which in some sense says that given any starting point (taken to be $\mathbf{0} \in \Z^3$ without loss of generality) there is a point which is (about) $k$ steps away in some cardinal direction and less than $k$ steps in the other two cardinal directions such that the eigenfunction mass at the latter point is at most exponentially small in $k$ compared to the eigenfunction mass at the former point. It is hence a very weak version of unique continuation. Formulating this cone lemma requires us to define exactly what is meant by cones in the present context, and notation for the layers as well.

Indeed, for any $\tau \in \{1,2,3\}$ and point $n \in \Z^3$ we let
\begin{equation}
    \mathrm{Co}^\tau_n := \{ m \in \Z^3\,:\, |(m-n)\cdot e_\tau| \geq \sum_{\varsigma \neq \tau} |(m-n)\cdot e_\varsigma|\}
\end{equation}
denote what we call henceforth the $\tau$ direction cone centered at $n$. For $k \in \Z$, we let
\begin{equation}
    \mathrm{Co}^\tau_n(k) := \{ m \in \mathrm{Co}^\tau_n\,:\, (n-m)\cdot e_\tau = k\}
\end{equation}
denote what we call the $k$-th layer of the cone $\mathrm{Co}^\tau_n$. The necessary cone lemma is then the following:

\begin{lem}[\cite{li2022anderson}]\label[lem]{conelem}
    Let $L \in \Z_+$, $K \in \R_+$ and $u,V$ be functions $\Lambda_L(\mathbf{0}) \rightarrow \R$ with $\|V\|_\infty \leq K$ solving $\Delta u = V u$ with Dirichlet boundary condition. For any $a \in \Lambda_L(\mathbf{0})$, $\tau \in \{1,2,3\}$, $\iota \in \{-1,+1\}$ and $k \in \Z_+$  such that $\mathrm{Co}^\tau_a(\iota k)$ intersects $\Lambda_L(0)$ non-trivially, there is point $a_k \in (\mathrm{Co}^\tau_a(\iota k) \cup \mathrm{Co}^\tau_a(\iota (k-1)))\cap \Lambda_L(0)$ with $|u(a_k)|\geq (K+11)^{-k}|u(a)|$.
\end{lem}

\section{Initial scale estimate}

In this section, we prove some technical lemmas which are non-stationary analogues of results from \cite{li2022anderson}. The proofs therein work mutatis mutandis, but we go through them explicitly both for the sake of completeness and to clarify dependence on new parameters present in the non-stationary setting. The first bounds the principal (least) eigenvalue of $H$ under the presumption that the potential is large at sites which are in some quantitative sense dense in a box; it is an analogue of \cite[Theorem B.1]{li2022anderson}. We only need the result for $d=3$, but we present a result valid in all dimensions at least three; some tweaks make the result work in dimension two, and indeed it is similar to \cite[Lemma 7.4]{ding2020localization}.

\begin{lem}[\cite{ding2020localization,li2022anderson}]\label{lifshitz}
    Fix a dimension $d \geq 3$ and $0 < \kappa \leq M$. Also fix an integer $L$ and another integer $R > 0$, the latter sufficiently large.  Let $V:\Lambda_L(\mathbf{0})\rightarrow \R$ be such that the set $\{n \in \Lambda_L(0)\,:\, V(n) \geq \kappa \}$ is an $R$-net in $\Lambda_L(0)$. Then the least eigenvalue $E_0$ of the operator $H_\Lambda$ (the restriction of $H = -\Delta + V$ to $\Lambda$ with Dirichlet boundary conditions) satisfies $E_0 \geq c_{\kappa,d} R^{-d}$.
\end{lem}
\begin{remm} For the special case $\kappa = M = 1$, the dimension 2 version of this statement is proved in the course of proving \cite[Lemma 7.1]{ding2020localization}, and in general dimensions this is \cite[Lemma B.1]{li2022anderson}. We prove only a slight generalization, which uses fundamentally the same proof.
\end{remm}

\begin{proof}
    By positivity of the potential, the principal eigenvalue $E_0$ satisfies
    \begin{equation*}
        E_0 = \sup_{u_0:\Lambda_L \to \R} \min_{\Lambda_L} \frac{Hu_0(n)}{u_0(n)}
    \end{equation*}
    and so the lower bound can be obtained by producing a function $u_0$ such that $\min_{\Lambda_L} \frac{Hu(n)}{u(n)} \geq CR^{-d}$. We let $G$  be the Green's function for the simple random walk on $\Z^d$ so that in particular $-\Delta G = \delta_{\mathbf{0}}$ and $0 \leq G(a) \leq G(\mathbf{0})$ for all $a \in \Z^3$. Moreover, we have the following asymptotics (see e.g. \cite[Theorem 4.3.1]{lawler2010random}):
    \begin{equation*}
        G(a) = \frac{C_d}{|a|^{d-2}} + O\left(\frac{1}{|a|^d}\right)
    \end{equation*}
    so that in particular
    \begin{equation*}
        .9\frac{C_d}{|a|^{d-2}} \leq G(a) \leq 1.1\frac{C_d}{|a|^{d-2}}
    \end{equation*}
    for $|a|$ sufficiently large. We first define an auxiliary function 
    \begin{align*}u(a) = \kappa^{-1} + G(\mathbf{0})-G(a) -\ve_dR^{-d}|a|^2
    \end{align*}for $\ve_d$ to be determined. We compute:
    $-\Delta u = -\delta_{\mathbf{0}} + 2d \ve_dR^{-d}$. By taking $\ve_d$ small enough, we can guarantee
    \begin{equation}
        0 < u(a) \leq \kappa^{-1} + G(\mathbf{0})
    \end{equation}
    for $|a|<3R$. In the regions $|a| < R$ and $2R < |a| < 3R$ we also have the following bounds:
\begin{align*}
        &u(a) < \kappa^{-1} + G(\mathbf{0}) -.9C_d R^{-d+2}&&\text{ if }\quad |a| < R\\
        &u(a) > \kappa^{-1} + G(\mathbf{0}) - 1.1C_d2^{-d+2} R^{-d+2} - 9\ve_d R^{-d+2} &&\text{ if }\quad 2R < |a| < 3R
\end{align*}
Then as long as $\ve_d \leq \frac{1}{9} [.9 - 1.1(2^{-d+2})]C_d$, we get
\begin{equation}\label{qmaxp}
    \min_{2R < |a| < 3R} u(a) > \max_{|a|<R}u(a)
\end{equation}
(In particular, for $d=3$,  $\ve_3 < .05C_3$ suffices here.) Now we define $u_0:\Lambda_L \to \R^+$ by
\begin{equation*}
    u_0(a) = \min\{u(a-b)\,:\,|b-a|< 3R\text{ and } V(b) \geq \kappa\}
\end{equation*}
Note first of all that this is well-defined because $\{b\,:\,V(b) \geq \kappa\}$ is an $R$-net. Moreover, for any arbitrary $a'$ we can find $b'$ with $|b'-a'|<2R$ such that $u_0(a') = u(a'-b')$. (Indeed, this follows from \eqref{qmaxp} together with the fact that $\{b\,:\,V(b) \geq \kappa$ is an $R$-net.) At the same time, if $a''$ is an immediate neighbor of $a'$, then $|a''-b'| < 3R$ and so we get
\begin{equation*}
    u_0(a'') \leq u(a''-b')
\end{equation*}
Now we compute (recall that $a'' \sim a'$ is the ``nearest neighbors'' relation on $\Z^d$, i.e. $\|a''-a'\|_1 = 1$):
\begin{align*}
    Hu_0(a') &= 2du_0(a') - \sum_{\substack{a'' \in \Lambda_L\\a'' \sim a'}} u_0(a'') + V(a')u_0(a')\\
    &\geq 2du(a'-b') - \sum_{\substack{a'' \in \Lambda_L\\a'' \sim a'}}u(a''-b') +V(a')u(a'-b')\\
    &= -\Delta u(a - b')|_{a = a'}+V(a')u(a'-b')\\
    &= \delta_{\mathbf{0}}(a'-b') + 2d\kappa^{-1}\ve_d R^{-d} + V(a')u(a'-b')\\
    &\geq 2d\kappa \ve_d R^{-d}
\end{align*}
with the last step using that $u(\mathbf{0}) = \kappa^{-1}$.
Finally, as $u_0(a') \leq \kappa^{-1}+G(\mathbf{0})$, we have $\min \frac{Hu_0(n)}{u_0(n)} \geq \frac{2d\ve_dR^{-d}}{\kappa^{-1}+G(\mathbf{0})}$, and we can take $c_{\kappa,d} = \frac{2d\ve_d}{\kappa^{-1} + G(\mathbf{0})}$; while this is not crucial for our argument, we point out that for $d$ fixed and $\kappa \to 0$, $c_{\kappa,d}$ is of order $\kappa$.
\end{proof}

This has as an immediate corollary the following, which is a parametrized version of \cite[Corollary B.3]{li2022anderson} and follows from the proof therein:
\begin{cor}[\cite{li2022anderson}]\label{initialscale}
    Take $H$ satisfying the assumptions (and thus conclusions) of \Cref{lifshitz}. For $R$ sufficiently large and any $\lambda$ satisfying $0 \leq \lambda \leq \frac{C_{\kappa,d}R^{-d}}{2}$, there is a constant $c> 0$ (which depends only $d$, $M$ and $\kappa$) such that
    \begin{equation}
        (H-\lambda)^{-1}(a,b) \leq cR^{d}\exp(-cR^{-d}|a-b|)
    \end{equation}
\end{cor}

We present their proof, mutatis mutandis:
\begin{proof}
For readability, we drop the dependence of $C_{\kappa,d}$ on $\kappa$ and $d$ in the notation. By \Cref{lifshitz}, $\|(H-\lambda)^{-1}\| \leq 2cR^{d}$; in particular $\sigma(H-\lambda) \subset [CR^{-d}/2,4d+M]$. The operator $T:=I - \frac{1}{4d+M}H$ has spectrum contained in $\left[0, 1- \frac{CR^{-d}}{8d+2M}\right]$. As $(H-\lambda)^{-1} = \frac{1}{4d+M}(I-T)^{-1}$, by Neumann series we get:
\begin{equation}
    (H-\lambda)^{-1}(a,b) = \frac{1}{4d+M}\sum_{k=0}^\infty T^i(a,b)
\end{equation}
Noting that $\|T\| \leq 1 - \frac{CR^{-d}}{8d+2M}$ and that $T^i(a,b) = 0$ if $|a-b| < i$, we get
\begin{align*}
    |(H-\lambda)^{-1}|(a,b) &\leq \frac{1}{4d+M}\sum_{i\geq|a-b|} \|T\|^i\\
    &\leq \frac{1}{4d+M}\frac{8d+M}{CR^{-d}}\left(1-\frac{CR^{-d}}{8d+2M}\right)^{-|a-b|}\\
    &= cR^d(1-C_{d,M}R^{-d})^{-|a-b|}
\end{align*}
Finally, for $R$ sufficiently large, $C_{d,M}R^{-d}$ is small enough that $\log( 1- C_{d,M} R^{-d}) \geq -2C_{d,M}R^{-d}$, so that (changing our constant), we get the desired bound. 
\end{proof}

\section{The inductive Wegner lemma}
The following inductive estimate is the main technical lemma necessary to establish Green's function estimates in the non-stationary setting; it establishes (under the presumption of a weak form of localization at smaller scales) that one is unlikely to find an eigenvalue of the finite volume operator in any given interval. It is a non-stationary version of \cite[Lemma 3.5]{li2022anderson}; the proof also has much in common with \cite[Lemma 5.1]{hurtado2026localization}, \cite[Lemma 5.6]{ding2020localization}.
\begin{lem}\label{bigweg}
Fix $M,\sigma^2 > 0$; all constants in this lemma are allowed to depend on $M$ and $\sigma^2$; we omit this in the notation for readability.  There is a constant $\ve_0$ such that if
\begin{enumerate}[label = (\Roman*)]
    \item $V = (V_n)_{n\in \Z^3}$ is a jointly independent random potential supported on $[0,M]^{\Z^3}$ with $\var V_n \geq \sigma^2$ for all $n$
    \item $\ve > \delta > 0$ with $\ve \leq c$
    \item $\overline{E} \in [0,12+M]$ is fixed
    \item $N$ is an integer and $\vec{\ell}$ is $\ve$-geometric for some $\ve > 0$
    \item $L_0 > L_1 > \cdots > L_4 > L_5 \geq C_{\ve, N,M,\sigma^2}$ are six dyadic scales satisfying
    \begin{equation}
        L_j^{1-2\delta} \geq L_{j+1} \geq L_j^{1 -\ve/2}
    \end{equation}
    \item $\Lambda \subset \Z^3$ is an $L_0$-cube
    \item $\Lambda_j'$, for $j$ ranging from $1$ to $N$, are ``defect'' $L_3$ cubes
    \item $G \subset \cup_{j=1}^N \Lambda_j'$ and $0 < |G| \leq L_0^\delta$
    \item $F \subset \Z^3$ is the intersection with $\Lambda$ of some $(CN, \vec{\ell}, \frac{1}{\ve}, \ve)$-graded set with first length scale at least $C_{\ve,\delta, N}$
    \item for any $L_3$-cube $\Lambda' \subset \Lambda \setminus \cup_{k=1}^N \Lambda_k'$, $F$ is $(1,\ve)$ normal
    \item \label{qlcond} For any potential $V:\Z^3 \to [0,M]$, letting $\mathcal{V}:= V|_{F\cap \Lambda}$, and taking $E$ with $|E-\overline{E}| \leq e^{-L_5}$ and $H_\Lambda u = \lambda u$, we have
    \begin{equation}
        e^{L_4} \|u\|_{\ell^2(\Lambda\setminus \cup_k \Lambda_k')} \leq \|u\|_{\ell^2(\Lambda)} \leq (1+L_0^{-\delta})\|u\|_{\ell^2(G)}
    \end{equation}
\end{enumerate}
then, letting $\sigma(F)$ denote the $\sigma$-algebra generated by $(V_n)_{n\in F}$, we have the bound
    \begin{equation}\label{bigwegbound}
        \essinf\P\left[\|(H_\Lambda - \overline{E})^{-1}\| \leq e^{L_1}\,|\,\sigma(F)\right] \geq 1 - L_0^{C\ve - \ve_0}
    \end{equation}
Moreover, with $C$ and $c$ denoting universal constants, the dependencies of various quantities in the hypotheses can be taken as follows
\begin{itemize}
    \item $\ve \leq\frac{1}{100}\min\{p-\frac{3}{2},1\}$, where $p$ is the constant from \Cref{ucfinal}
    \item $L_5 \geq \max\{C^{\delta^{-1}},\log\left(\frac{\sigma^{10}}{4M^9}\right),\left(\frac{\sigma^5}{2M^4}\right)^{-5/2\ve}, 2\log(2M)\}$
\end{itemize}

\end{lem}
\begin{remm}
    As $M$, $\sigma^2$, $\ve$ and $\delta$ will not vary over the course of the MSA argument, the asymptotics of optimal $\ve$, $L_5$ do not matter for our localization results; we record the requisite largeness/smallness of $L_5$ and $\ve$ respectively in a quantitative sense in case it proves useful in the future.
\end{remm}
Except for our use of Bernoulli decompositions and the need for stronger combinatorial estimates, our proof will follow that of \cite[Lemma 3.5]{li2022anderson} quite closely. Our approach also carefully follows that presented in \cite{hurtado2026localization}; we have tried to further clarify the arguments therein here. We require a technical lemma \cite[Lemma 5.1]{ding2020localization} regarding eigenvalue variation:
\begin{lem}[\cite{ding2020localization}]\label{eigenvar}
Let $A$ be a symmetric $n\times n$ matrix with eigenvalues $\lambda_1 \geq \cdots \geq \lambda_n$ and orthonormal eigenbasis $v_1,\dots,v_n \in \R^n$. We also let $e_1,\dots, e_n$ denote the standard orthonormal basis, and $e_k \otimes e_k$ the rank one projection onto $e_k$. If
\begin{enumerate}
    \item\label{eigenvar1} $0 < r_1 < r_2 < r_3 < r_4 < r_5 < 1$
    \item\label{eigenvar2} $r_1 \leq c \min\{r_3r_5,r_2r_3/r_4\}$
    \item\label{eigenvar3} $0 < \lambda_j \leq \lambda_i < r_1 < r_2 < \lambda_{i-1}$
    \item\label{eigenvar4} $|\langle v_j, e_k\rangle|^2\geq r_3$
    \item\label{eigenvar5} $\sum_{r_2<\lambda_\ell < r_5} |\langle v_\ell,e_k\rangle|^2 \leq r_4$
\end{enumerate}
    with $c> 0$ a universal constant, then $\tr 1_{[r_1,\infty)}(A) < \tr 1_{[r_1,\infty)}(A+e_k\otimes e_k)$; i.e. adding the rank one projection $e_k \otimes e_k$ ``pushes'' the $i$-th eigenvalue over $r_1$.
\end{lem}
To be quite precise, we will make use of the following easy corollary, which follows from this lemma via basic eigenvalue variation.
\begin{cor}\label{eigenvarcor}
    Under all the same hypotheses, for any $\eta \geq 1$, $\tr 1_{[r_1,\infty)}(A) < \tr 1_{[r_1,\infty)}(A+\eta (e_k\otimes e_k))$.
\end{cor}

We will also need a linear algebraic result giving an upper bound on ``almost orthonormal'' collections of vectors. The following is a variation on \cite[Lemma 5.2]{ding2020localization}, see also \cite{taobound}. The proof of \cite[Lemma 5.2]{ding2020localization} suffices to prove this variation.

\begin{lem}[\cite{ding2020localization,taobound}]\label{taolem}
    Fix $\alpha > 0$. Let $v_1,\dots,v_m \in \R^n$ be such that $|\langle v_j, v_i \rangle - \delta_{ij}| \leq \alpha n^{-1/2}$. If $\alpha n \geq 1/2$
    \begin{equation*}
        m \leq \left(\frac{\alpha^2 - \alpha}{2}\right)n
    \end{equation*}
\end{lem}

Having introduced all the necessary results, we prove our main technical result; we reiterate that most of this proof proceeds along the lines of the proof of \cite[Lemma 3.5]{li2022anderson}, with the main novelty being the use of Bernoulli decompositions. In particular, there is no significant novelty compared to the work of Li and Zhang until the proof of \Cref{annuluseventclaim}; from this point onwards the use of Bernoulli decompositions and combinatorial bounds from \cite{hurtado2026localization} is necessary. We include the proofs of various technical claims which all originate in \cite{li2022anderson} for the sake of completeness. The general approach is also largely similar to that used by the author in proving \cite[Lemma 5.1]{hurtado2026localization}.

\begin{proof}[Proof of \Cref{bigweg}]
    We immediately presume $\ve < .01$ and $L_5^{\delta} > 100$; assuming these to be small in a universal way gives us the freedom to use basic bounds coming from linearizing without introducing too much notational clutter. Hence we will use bounds like $1+\ve \leq \frac{1}{1-\ve} \leq 1+2\ve$, $L_j \leq L_{j+1}^{1.02}$ freely. Throughout the proof, we will let $C$ and $c$ represent large and small constants respectively, which may change line to line. When such constants depend on parameters, we will make this dependence explicit via a subscript, e.g. $C_\ve$ for a large constant depending on $\ve > 0$.

    We let $E_1 \geq \dots \geq E_{(L_0+1)^3}$ be the eigenvalues of $H_\Lambda$ listed in decreasing order with multiplicity; naturally these are random variables depending on the random potential restricted to $\Lambda$. We let $u_k \in \ell^2(\Lambda)$ be an eigenbasis; we fix some deterministic scheme for choosing an orthonormal basis so that in particular the vectors $u_k$ are measurable with respect to the potential restricted to $\Lambda$.

    We let $F' = F \cup  (\cup_{1 \leq j \leq N} \Lambda_j')$. If we let $\sigma(F')$ denote the $\sigma$-algebra generated by $(V_n)_{n\in F'}$, clearly for any event $\mathcal{E}$ we have:
    \begin{equation*}
        \essinf \P[\mathcal{E}\,|\,\sigma(F)] \geq  \essinf \P[\mathcal{E}\,|\,\sigma(F')]
    \end{equation*}
    Hence, it suffices to prove:
    \begin{equation}\label{bigwegbound2}
        \P[\|(H_\Lambda - \overline{E})^{-1}\| \leq e^{L_1}\,|\, \sigma(F')] \geq 1 - L_0^{C\ve - \ve_0}
    \end{equation}
    Thus, the rest of the proof is aimed toward the end of proving this stronger estimate. To this end, we need the existence for any eigenfunction $u$ of some site sufficiently far away from defect boxes carrying a non-trivial amount of mass for $u$.
    \begin{claim}
       There is a constant $C_N$ such that for any $\lambda \in [0,12+M]$ and $u \in \ell^2(\Lambda)$ solving $H_\Lambda u = \lambda u$, there exists $a' \in \Lambda$ with $\Lambda_{\frac{L_3}{2}}(a') \subset \Lambda\setminus \cup_{1\leq j \leq N} \Lambda_j'$ and moreover
        \begin{equation}
            |u(a')| \leq e^{-C_NL_3}\|u\|_{\ell^\infty(\Lambda)}
        \end{equation}
    \end{claim}
    We take $a_0 \in \Lambda$ to solve $|u(a_0)| = \|u\|_{\ell^\infty(\Lambda)}$. Without loss of generality, we presume for the proof of this claim that $\Lambda$ is centered at $0 \in \Z^3$ and that all of the coordinates of $a_0$ are non-positive, i.e. $a_0 \cdot e_i \leq 0$ for $i = 1,2,3$. Assuming our scales $L_i$ are sufficiently large (how large will be clarified shortly), we obtain by the pigeonhole principle that there is some integer $x_0' \in [a_0\cdot e_1 + 100 NL_3, a_0\cdot e_1 + 200NL_3]$ so that
    \begin{equation*}
        \{b \in \Lambda\,:\, b\cdot e_1 \in [x_0'-16L_3, x_0'+16L_3]\}\cap (\cup_{1 \leq j \leq N} \Lambda_j') = \varnothing
    \end{equation*}
    (One can think of this as finding $x_0'$ which is far away from the shadows cast on $\Z$ by projecting all the defect boxes onto the first coordinate.) Indeed, the sites $x_0'$ which we cannot choose are on the order of $38NL_3$ at worst, each defect cube carving out these many; it is only necessary to ensure that $[a_0\cdot e_1 + 100 NL_3, a_0\cdot e_1 + 200NL_3]$ is a subset of $[-L_0, L_0]$.

    By assumption, $L_3 \leq L_0^{(1-2\delta)^3}$, so it suffices that $L_0^{1-(1-2\delta)^3} > 200$. Since $L_0 \geq L_5$ and sharp estimates on the lower bounds for $L_5$ are not crucial, it suffices to take
    \begin{equation*}
        L_5 > 200^{(1-(1-2\delta)^3)^{-1}} \geq 200^{1+16\delta} \geq 40000
    \end{equation*}
    which is already assumed.

    Now we use the cone lemma (\Cref{conelem}) to produce the appropriate $a'$.
    Specifically, we let $a_1 \in C^1_{a_0}(x_0' - a_0\cdot e_1) \cup C^1_{a_0}(x_0' - a_0\cdot e_1 + 1) \cap \Lambda$ with $|u(a_1)| \geq (11+M)^{-200N}|u(a_0)|$. If $a_0$ is not too close to the boundary of $\Lambda$ we are done; in particular $\Lambda_{L_3/2}(a_1)$ does not intersect any of the defect boxes. However, if $a_0$ is sufficiently close to the boundary, said box around it may not fall inside $\Lambda$.

    Hence, we use two more applications of the cone lemma to account for the cases where $a_1$ is near the boundary of $\Lambda$. By our assumption that $a_1$ is in the octant with $a_0 \cdot e_i \leq 0$ for $i \in \{1,2,3\}$, if $a_1$ is near the boundary it is near at least one of the boundary faces $\{n \in \Lambda\,:\, n\cdot e_2 = - L_0\}$ or $\{n \in \Lambda\,:\, n\cdot e_3 = -L_0\}$. (By moving quite far in the $e_1$ direction from $a_0$, $a_1$ cannot be near the boundary face $\{n \in \Lambda\,:\, n\cdot e_1 = -L_0\}$. Hence we can take $a_2 \in C^2_{a_1}(4L_3)\cup C^2_{a_1}(4L_3+1)$ such that $|u(a_2)| \geq (11+M)^{-4L_3+1} |u(a_1)|$ and then $a_3 \in C^3_{a_2}(2L_3) \cap C^3_{a_2}(2L_3+1)$ such that $|u(a_3)| \geq (11+M)^{-2L_3+1}|u(a_2)|$, using the cone lemma.

    Note that $\Lambda_{L_3/2}(a_3)$ intersects all defect cubes trivially; $x_0' + 6L_3 + 2\leq a_3 \cdot e_1 \leq x_0' + 6L_3 + 2$ the projection onto the first coordinate of $\Lambda_{L_3/2}(a_3)$ intersects the projection onto the first coordinate of all defect boxes trivially, implying trivial intersection for the original sets. At the same time, $\Lambda_{L_3}(a_3) \subset \Lambda$ essentially because we have walked far enough in the positive direction to preclude being near the boundary. Hence the claim follows with $C_N = \log(11 +M)(N + 3)$ and $a' = a_3$.

    From this and \Cref{ucfinal}, we obtain the following more or less immediately:
    \begin{claim}\label{ucclaim}
        For any $E \in [0,12+M]$, $H_\Lambda u = E u$ implies
        \begin{equation}
            \left|\left\{n \in \Lambda\,:\,|u(n)| \geq e^{-L_2/2}\|u\|_{\ell^2(\Lambda)}\right\}\setminus F'\right| \geq \left( \frac{L_3}{2}\right)^p
        \end{equation}
        where $p >3/2$ is precisely the same constant appearing in \Cref{ucfinal}.
    \end{claim}

    Indeed, by the previous claim, one can find $a'$ such that $\Lambda_{L_3/2}(a') \subset \Lambda \setminus \
    \cup_j \Lambda_j'$ and moreover
    \begin{equation*}
        |u(a')| \geq (11+M)^{C_NL_3}\|u\|_{\ell^\infty(\Lambda)} \geq (L_0+1)^3(11+M)^{C_NL_3} \|u\|_{\ell^2(\Z^3)}
    \end{equation*}
    In particular, since we have presumed that $F$ is normal in all $L_3$ subcubes trivially intersecting the defect cubes, we can apply \Cref{ucfinal} on the cube $\Lambda_{L_3/2}(a')$, with graded set $E$, and $\|V\|_\infty \leq M$ to obtain the result, so long as $L_3$ and the first length scale $\ell_1$ are sufficiently large, which is imposed by our hypotheses.

    The next claim bounds the probability of having a specific ``spectral gap'' over an annulus. Specifically, we let $s_i = \exp(-L_1 + (L_2-L_4+C)i)$; note that these are very small quantities for $i \ll L_1/L_2$, but they grow exponentially in $i$. For integers $1 < k_1 \leq k_2 < (L_0+3)^3$ we let $\mathcal{E}_{k_1,k_2,\ell}$ denote the event that
    \begin{enumerate}
    \item $|E_{k_1} - \overline{E}| \leq s_\ell$ and $|E_{k_2} - \overline{E}| \leq s_\ell$
    \item $|\lambda{k_1-1} - \overline{E}| \geq s_{\ell+1}$ and $|E_{k_2} -\overline{E}| \geq s_{\ell+1}$
    \end{enumerate}

    Note that the union event $\cup_{1 < k_1 \leq k_2 < (L_0+1)^3} \mathcal{E}_{k_1,k_2,\ell}$ is precisely the event that the annulus $[\overline{E} - s_{\ell+1}, \overline{E}+s_{\ell+1}] \setminus [\overline{E} - s_\ell, \overline{E}+s_\ell]$ contains no eigenvalues of $H_\Lambda$ but all three connected components of the complement each contain at least one eigenvalue. See \Cref{annuluseventfig} for a pictorial depiction.
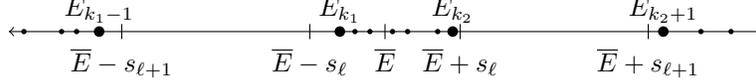
\begin{figure}[h]

\centering
\begin{tikzpicture}
	\draw[->] (5,0) -- (10,0);
    \draw[->] (5,0) -- (0,0);
    \draw (8.5,.1) -- (8.5,-.1) node[anchor = north] {$\overline{E} + s_{\ell+1}$};
    \draw (1.5,.1) -- (1.5,-.1) node[anchor = north] {$\overline{E} - s_{\ell+1}$};
    \draw (5,.1) -- (5,-.1) node[anchor=north] {$\overline{E}$};
    \draw (4,.1) -- (4,-.1) node[anchor = north] {$\overline{E} - s_\ell$};
    \draw (6,.1) -- (6,-.1)node[anchor = north] {$\overline{E} + s_\ell$};
    \fill (1.2,0) circle (2 pt) node [anchor = south] {$E_{k_1-1}$};
    \fill (4.4, 0) circle (2 pt) node [ anchor = south] {$E_{k_1}$};
    \fill (5.9, 0) circle (2 pt) node [anchor=south] {$E_{k_2}$};
    \fill (8.7, 0) circle (2 pt) node [anchor=south] {$E_{k_2+1}$};
    \fill (4.6, 0) circle (1 pt);
    \fill (4.8, 0) circle (1 pt);
    \fill (5.1, 0) circle (1 pt);
    \fill (5.3, 0) circle (1 pt);
    \fill (5.7, 0) circle (1 pt);
    \fill (0.9, 0) circle (1 pt);
    \fill (0.7, 0) circle (1 pt);
    \fill (0.2, 0) circle (1 pt);
    \fill (9.2, 0) circle (1 pt);
    \fill (9.6, 0) circle (1 pt);
\end{tikzpicture}
\caption{A possible eigenvalue configuration for the event $\mathcal{E}_{k_1,k_2,\ell}$}\label{annuluseventfig}
\end{figure}

The next important claim is that these events are improbable:
\begin{claim}\label{annuluseventclaim}
    For $1 < k_1 \leq k_2 < (L_0+1)^3$ and $\ell \leq CL_0^\delta$, we have
    \begin{equation}\label{annulusbound}
        \esssup \P[\mathcal{E}_{k_1,k_2,\ell}\,|\,\sigma(F')] \leq C_{M,\sigma}L_0^{3/2}L_3^p
    \end{equation}
\end{claim}

In order to prove this claim, it is necessary to decompose our potential. Specifically, by \Cref{berdecomp}, there are for all $n \in \Z^3$ measurable functions $Y_n,\,Z_n:(0,1) \to \R$ and $p_n \in (0,1)$ such that
\begin{enumerate}
    \item $\iota:=\inf_{\substack{s \in (0,1)\\n \in \Z^3}} Z_n(s) > 0$
    \item $p_-:=\inf_{n\in\Z^3} p_n > 0$
    \item $p_+:=\sup_{n\in\Z^3} p_n < 1$
    \item $M':=\sup_{\substack{s \in (0,1)\\n \in \Z^3}} \max\{|Y_n(s)|,|Z_n(s)|\} < \infty $
    \item $V_n \overset{\mathcal{D}}{=} Y_n(t) + Z_n(t) \xi$, where $t$ and $\xi$ are independent variables, with $t$ uniformly distributed on $(0,1)$ and $\xi$ Bernoulli with $\P[\xi = 1] = 1 - \P[\xi= 0] = p_n$
\end{enumerate}

\begin{remm}
    $M'$ is necessarily less than or equal to $M$, where $M$ is the absolute bound on the magnitude of our potential; we will use $M$ as an upper bound on $|X_n(t_n)|$ and $|Y_n(t_n)|$ going forward.
\end{remm}
In particular, we consider an identically distributed system where we let $(t_n)_{n\in\Z^3}$ be uniformly distributed on $(0,1)$, $(\xi_n)_{n\in\Z^3}$ be Bernoulli variables with $\P[\xi_n = 1] = 1-\P[\xi_n=0] = p_n$, all jointly independent, and $V_n = Y_n(t_n) + Z_n(t_n)\xi_n$. After making what is in some sense a change of variables, the potential (and hence operator) restricted to $\Lambda$ is completely determined by the values of $t_n$ for $n \in \Lambda$ and $\xi_n$ for $n \in \Lambda$.

For the system so constituted, we can now use combinatorial estimates, by looking at the variables $\xi_n$. We will use boldface notation as a shorthand for (random) vectors in $\R^{\Lambda\setminus F'}$. For example, we let $\mathbf{t}$ denote $(t_n)_{n \in \Lambda\setminus F'}$, and $\boldsymbol{\xi}$ denotes $(\xi_n)_{n\in\Lambda\setminus F'}$. The sought bound \Cref{annulusbound} will follow a stronger conditional version, which we first formulate more abstractly.

Let $\sigma_u(F')$ be the $\sigma$-algebra generated by $(t_n)_{n\in\Lambda\setminus F'}$ and $\sigma(F')$. Then, it will suffice now to show that:
\begin{equation}\label{annulusbound2}
    \esssup \P[\mathcal{E}_{k_1,k_2,\ell}\,|\, \sigma_u(F')] \leq C_{M,\sigma}L_0^{3/2}L_3^p
\end{equation}

While we have generally adopted a more abstract perspective using $\sigma$-algebras, actually proving this bound requires consideration of the system more explicitly. In particular, note that the conditional probability above is only guaranteed to exist almost everywhere by general theory, but is in fact regular and defined everywhere due to the nature of the problem. Indeed, the conditional probability here at specific deterministic values $\mathbf{t} = \mathbf{t}' :=(t_n')_{n\in\Lambda \setminus F} \in (0,1)^{\Lambda\setminus F'}$ and $V|_{F'} = \mathcal{V}' \in [0,M]^{F'}$ is precisely the probability of $\mathcal{E}_{k_1,k_2,\ell}$ when one fixes those values and allows $\boldsymbol{\xi}$ to vary randomly according to its distribution.

Using our combinatorial bounds will require that we identify $\boldsymbol{\xi}$ with a random set valued variable in $2^{\Lambda \setminus F'}$ by saying that $n \in \Lambda\setminus F'$ if and only if $\xi_n = 1$. Then one can define the event $\mathcal{E}_{k_1,k_2,\mathbf{t}',\mathcal{V}'}$ depending on $\boldsymbol{\xi}$ only, with sample space $2^{\Lambda\setminus F'}$ which holds if and only if $\mathcal{E}_{k_1,k_2,\ell}$ holds for the operator $H_\Lambda$ with $V|_{F'} = \mathcal{V}'$ and $V_n = Y_n(t_n') + Z_n(t_n')\xi_n$. We have

\begin{equation}\label{finalannbound}
    \P[\mathcal{E}_{k_1,k_2,\ell}\,|\,\sigma_u(F')](\mathbf{t}',\mathcal{V}') = \P[\mathcal{E}_{k_1,k_2,\ell,\mathbf{t}',\mathcal{V}'}]
\end{equation}
so in particular we make one final reduction of \Cref{annuluseventclaim}, to proving an upper bound on the right hand side of (\ref{finalannbound}).

In order to prove this bound, we define two more events, $\mathcal{E}_{k_1,k_2,\ell, \mathbf{t}',\mathcal{V}',i}$ for $i=0,1$. $\mathcal{E}_{k_1,k_2,\ell,\mathbf{t}',\mathcal{V}',i}$ holds if and only if $\mathcal{E}_{k_1,k_2,\ell,\mathbf{t}',\mathcal{V}'}$ holds and moreover
\begin{equation*}
    \left|\left\{ n \in \Lambda \setminus F'\,:\, |u_{k_1}(n)| \geq e^{-L_2/4} \text{ and } \xi_n = i \right\}\right| \geq \frac{L_3^p}{16}
\end{equation*}
keeping in mind that we have fixed the value of $\mathbf{t}'$ and $\mathcal{V}'$ in the event, and note that by \Cref{ucclaim} and $3 \geq p > 3/2$, we get
\begin{equation*}
    \mathcal{E}_{k_1,k_2,\ell, \mathbf{t}',\mathcal{V}'} \subset \mathcal{E}_{k_1,k_2,\ell,\mathbf{t}',\mathcal{V}',0} \cup \mathcal{E}_{k_1,k_2,\ell,\mathbf{t}',\mathcal{V}',1}
\end{equation*}

We then (for a fixed choice of $i = 0,1$) define two subsets of $\Lambda$ depending on $\boldsymbol{\xi}$:
\begin{align*}
    S_1(\boldsymbol{\xi})&:= \{n \in \Lambda \setminus F'\,:\, \xi_n = 1-i\}\\
    S_2(\boldsymbol{\xi})&:= \{n \in \Lambda \setminus F'\,:\, \xi_n = i \text{ and } |u_{k_1}(n)| \geq e^{-L_2/4}\}
\end{align*}

For $n \in \Lambda\setminus F'$, we denote by $\boldsymbol{\xi}^n$ the ``$n$-flip'' of $\boldsymbol{\xi}$, defined by
\begin{equation*}
    (\boldsymbol{\xi}^{n})_m = \begin{cases}
        \xi_m &\text{ if } m\neq n\\
        1-\xi_m&\text{ if } m = n
    \end{cases}
\end{equation*}

We will now demonstrate that if $\boldsymbol{\xi} \in \mathcal{E}_{k_1,k_2,\ell,\mathbf{t}',\mathcal{V}',0}$, then $\boldsymbol{\xi}^n \notin \mathcal{E}_{k_1,k_2,\ell,\mathbf{t}',\mathcal{V}',0}$ for all $n \in S_2(\boldsymbol{\xi})$, and hence $\mathcal{E}_{k_1,k_2,\ell,\mathbf{t}',\mathcal{V}',0}$ is in a certain sense fragile. (And a similar thing can be shown for $\mathcal{E}_{k_1,k_2,\ell,\mathbf{t}',\mathcal{V}',1}$ by a nearly identical argument, which we will not explicitly present.)

Indeed, this follows from an application of the eigenvariation result \Cref{eigenvar} for an appropriate chosen operator and radii. Note that said lemma concerns unit sized rank one perturbations; we let $H_{\Lambda,\mathbf{t}',\mathcal{V}'}(\boldsymbol{\xi})$ denote the random operator $H_\Lambda$ when $\mathbf{t}$ and $V|_{F'}$ are fixed at the values $\mathbf{t}'$ and $\mathcal{V}'$ respectively. Then (for $\boldsymbol{\xi} \in \mathcal{E}_{k_1,k_2,\ell,\mathbf{t}',\mathcal{V}',0}$ and $n \in S_2(\boldsymbol{\xi})$) we have $H_{\Lambda,\mathbf{t}',\mathcal{V}'}(\boldsymbol{\xi}^n) = H_{\Lambda,\mathbf{t}',\mathcal{V}'}(\boldsymbol{\xi}) + Z_n(t_n') (\delta_n \otimes \delta_n)$. We recall that there is a uniform positive lower bound $\iota \leq Z_n(t_n')$.

Hence we can apply \Cref{eigenvarcor} for the operator $\iota^{-1}\left(H_{\Lambda,\mathbf{t}',\mathcal{V}'}(\boldsymbol{\xi})-\overline{E} + s_\ell\right)$ and the radii $r_1 = 2\iota^{-1}s_\ell$, $r_2 = \iota^{-1}s_{\ell+1}$, $r_3 = \iota^{-1}e^{-L_2/2}$, $r_4 = \iota^{-1}e^{-cL_4}$ and $r_5 = \iota^{-1}e^{-L_5}$. Ensuring the hypotheses \ref{eigenvar1} and \ref{eigenvar2} for \Cref{eigenvar} concerning the relative sizes of the radii is just a question of taking $L_5$ sufficiently large, in a way which depends on $\delta$ and on $\iota$. Straightforward computations yield that (if one assumes $\ve < 1$) it suffices to take
\begin{equation*}
L_5 \geq \max\left\{ \log(\iota^{-1}),(100C_1)^{2\delta^{-1}}\right\}
\end{equation*}
Hypothesis \ref{eigenvar3} is a consequence of our assumption that $\boldsymbol{\xi} \in \mathcal{E}_{k_1,k_2,\ell,\mathbf{t}',\mathcal{V}',0}$. Hypothesis \ref{eigenvar4} follows from the fact that $n \in S_2(\boldsymbol{\xi})$; hypothesis \ref{eigenvar5} follows from hypothesis \ref{qlcond} in our Wegner lemma. Hence in particular, $E_{k+1} \geq \overline{E} + s_{\ell+1}$ if one flips $\boldsymbol{\xi}$ at any $n \in S_2(\boldsymbol{\xi})$.

Finally, note that in particular this means that for any $\boldsymbol{\xi} \in \mathcal{E}_{k_1,k_2,\ell,\mathbf{t}',\mathcal{V}',0}$, we have $\boldsymbol{\xi} \cup \{n\} \notin \mathcal{E}_{k_1,k_2,k\ell,\mathbf{t}',\mathcal{V}',0}$ for all $n \in S_2(\boldsymbol{\xi})$. Hence $\mathcal{E}_{k_1,k_2,\ell,\mathbf{t}',\mathcal{V}',0}$ is a $\kappa$-Sperner set for $\kappa = \frac{|S_2(\boldsymbol{\xi})|}{|\Lambda\setminus F'|} \leq \frac{L_3^p}{8L_0^3}$. If we let $\beta = \min\{p_-,1-p_+\}$, then by \Cref{kappaspern}, we immediately obtain:
\begin{align*}
    \P[ \boldsymbol{\xi} \in \mathcal{E}_{k_1,k_2,\ell,\mathbf{t}',\mathcal{V}',0}] &\leq C\beta^{-5/2}\left(\frac{L_3^{p}}{(2L_0+1)^3}\right)^{-1}\left(\frac{(2L_0+1)^3}{2}\right)^{-1/2}\\
    &\leq \frac{C L_0^{3/2}}{\beta^{5/2} L_3^p}
\end{align*}

By a similar argument, we get the same bound for $\mathcal{E}_{k_1,k_2,\ell,\mathbf{t}',\mathcal{V}',1}$ and hence (up to a factor of two absorbed into the universal constant) we also get the same bound for $\mathcal{E}_{k_1,k_2,\ell,\mathbf{t}',\mathcal{V}'}$, as was desired.

Having proven the claim, we aim to use union bounds to control $\P[\|(H_{\Lambda}-\overline{E})^{-1}\| > e^{L_1}\,|\,\sigma(F')]$; towards this end we would like to cover the event that $\|(H_\Lambda - \overline{E})^{-1}\|$ is large with events of the form $\mathcal{E}_{k_1,k_2,\ell}$. This requires us to consider values of the potential which may lie outside the essential support in order to carry out continuous eigenvalue variation arguments. In particular, as we have done so far, we let $\mathcal{V}'$ be some fixed, deterministic configuration of the potential on $F'$. Though it is not strictly necessary, in order to avoid introducing new notation, we also fix some deterministic value $\mathbf{t}'$ for $\mathbf{t}$. Then for any $\boldsymbol{\xi} \in [0,1]^{\Lambda \setminus F'}$ (note that we are considering not just the discrete hypercube, but all of the hypercube) one can define an operator $H_{\Lambda,\mathcal{V}',\mathbf{t}'}(\boldsymbol{\xi})$ which is the negative Laplacian plus a potential $V$ defined by
\begin{equation*}
    V_n = \begin{cases} \mathcal{V}'_n &\text{ if } n \in F' \\ Y_n(t_n') + Z_n(t_n')\xi_n &\text{ if } n \notin F'\end{cases} 
\end{equation*}
In this way, our operator naturally extends to values outside the essential support of the distribution. The final major technical claim in proving the lemma is the following:

\begin{claim}
    Let
    \begin{equation*}
        K = \{k \in \{1,\dots,(L_0+1)^3\,:\, \text{There is some } \boldsymbol{\xi} \in [0,1]^{\Lambda\setminus F'} \text{ such that } |E_k(\boldsymbol{\xi}) - \overline{E}| \leq e^{-L_2} \}
    \end{equation*}
    Then
    \begin{enumerate}
        \item As events on the sample space $2^{\Lambda\setminus F'}$, we have
        \begin{equation}
            \{\|(H_{\Lambda,\mathbf{t}',\mathcal{V}'}-\overline{E})^{-1}\| \geq e^{L_1}\} \subset \bigcup_{\substack{k_1,k_2\in K\\0 \leq \ell \leq CL_0^\delta}} \mathcal{E}_{k_1,k_2,\ell,\mathcal{V}',\mathbf{t}'}
        \end{equation}
        \item $|K| \leq 20L_0^\delta$
    \end{enumerate}
\end{claim}

First, let 
\begin{equation*}
    K'(\boldsymbol{\xi}) = \{k \in \{1,\dots,(L_0+1)^3\}\,:\,|\overline{E} - E_k(\boldsymbol{\xi})| \leq e^{-L_2} \}
\end{equation*}We let $k_1\leq \cdots \leq k_m$ be the elements of $K'$, with $m' = |K'|$.  We will treat the first part of the claim first. By e.g. the spectral theorem, $\|(H_{\Lambda,\mathbf{t}',\mathcal{V}'}(\boldsymbol{\xi}) - \overline{E})^{-1}\| > e^{L_1}$ if and only if $|E_k(\boldsymbol{\xi}) - \overline{E}| \leq e^{-L_1}$ for some $k$. Consider the sequence of annuli $I_\ell = [\overline{E} - s_{\ell+1}, \overline{E} + s_{\ell+1}]\setminus [\overline{E}-s_\ell, \overline{E}+s_\ell]$, for $\ell = 1, \dots, m$. Taking for granted currently the bound $m = |K'| \leq CL_0^\delta$, clearly $H_\Lambda$ has eigenvalues greater than $\overline{E} + e^{-L_2} \geq \overline{E} + s_{m+1}$ and lesser than $\overline{E} - e^{-L_2} \leq \overline{E} -s_{m+1}$. If moreover $H_\Lambda$ has at least one eigenvalue inside the range $[\overline{E}-e^{-L_1},\overline{E}+e^{L_1}]$, there are $m-1$ remaining eigenvalues and $m$ annuli; one will necessarily be left empty. Letting $\ell$ be the index corresponding to the annulus bereft of eigenvalues, and $E_{k_1}, E_{k_2}$ the least and greatest eigenvalues in $[\overline{E}-s_\ell, \overline{E}-s_\ell]$ respectively, we get that $\mathcal{E}_{k_1,k_2,\mathcal{V}',\mathbf{t}'}$ holds.

In particular, we have shown
\begin{equation}\label{contain423}
\{\|(H_{\Lambda,\mathbf{t}',\mathcal{V}'}(\boldsymbol{\xi})-\overline{E})^{-1}\| \geq e^{L_1}\} \subset \bigcup_{\substack{k_1,k_2\in K'\\0 \leq \ell \leq CL_0^\delta}}\mathcal{E}_{k_1,k_2,\ell,\mathcal{V}',\mathbf{t}'}
\end{equation}
On the other hand, clearly $K'(\boldsymbol{\xi}) \subset K$ for all $\boldsymbol{\xi}$ by definition. Hence the first part of the claim is proved, and it remains to demonstrate that $|K| \leq CL_0^\delta$. 

This proceeds in two steps; both rely on our weak localization hypothesis \Cref{qlcond}. The first is the following containment:
\begin{equation*}
    K \subset K'' :=\left\{ k \in \{1,\dots,(L_0+1)^3 \,: |\overline{E} - E_k(\boldsymbol{\xi})| \leq e^{-L_4} \text{ for all } \boldsymbol{\xi} \in [0,1]^{\Lambda\setminus F'}\right\}
\end{equation*}
Indeed, let $k$ and $\boldsymbol{\xi}' \in [0,1]^{\Lambda\setminus F'}$ be such that $|\overline{E}-E_k(\boldsymbol{\xi})| \leq e^{-L_2}$. Fixing another $\boldsymbol{\xi} \in [0,1]^{\Lambda \setminus F'}$, we let $\boldsymbol{\xi}_s := \boldsymbol{\xi}' + s(\boldsymbol{\xi} - \boldsymbol{\xi}')$ linearly interpolate between them. Then we have $H_{\Lambda,\mathcal{V},\mathbf{t}'}(\boldsymbol{\xi}_s) = H_{\Lambda,\mathcal{V},\mathbf{t}'}(\boldsymbol{\xi}) + s(\sum_{n \in \Lambda\setminus F'} Y_n(t'_n) (\xi_n - \xi_n')\delta_n \otimes \delta_n)$. (Recall that $v \otimes v$ denotes the rank one operator $ w \mapsto \langle v, w \rangle v$.) By standard eigenvalue variation using the Feynman-Hellmann formula,
\begin{align*}
    \frac{d}{ds} E_k(\boldsymbol{\xi}_s) &= \left\langle u_k(\boldsymbol{\xi}_s), \left(\frac{d}{ds} H_{\Lambda,\mathcal{V},\mathbf{t}'}(\boldsymbol{\xi}_s)\right) u_k(\boldsymbol{\xi}_s)\right\rangle\\
    &=\sum_{n \in \Lambda\setminus F'} Y_n(t_n') |\langle u_k(\boldsymbol{\xi}_s),\delta_n\rangle|^2
\end{align*}
Recall that $0 < \iota \leq  Y_n(t_n') \leq M$, so that $\|Y_n(t_n')(\xi_n - \xi_n')\|_{\ell^2(\Lambda \setminus F')} \leq M|\Lambda|^{1/2}$ and hence by Cauchy Schwarz 
\begin{equation*}
    \left|\frac{d}{ds} E_k(\boldsymbol{\xi}_s)\right| \leq M|\Lambda|^{1/2}\sum_{n \in \Lambda \setminus F'} |\langle u_k(\boldsymbol{\xi}),\delta_n\rangle|^2
\end{equation*}
Recalling that $|\overline{E} - E_k(\boldsymbol{\xi}')|\leq e^{-L_2}$ by assumption and using Hypothesis \ref{qlcond}, we get:
\begin{align*}
     |E_k(\boldsymbol{\xi}) - E_k(\boldsymbol{\xi}')| &\leq 2M|\Lambda|^{1/2}\sum_{n\in\Lambda\setminus F'}\int_0^1 |\langle u_k(\boldsymbol{\xi}_s),\delta_n\rangle|^2\,ds\\
     &\leq 2M |\Lambda|^{1/2}\int_0^1 e^{-2L_4} + 1_{|\overline{E} - E_k(\boldsymbol{\xi}_s)| \geq e^{-L_5}}\,ds
\end{align*}
where we are using the fact that $\|u_k\|_{\ell^2(\Lambda \setminus F')} \leq 1$ universally and $\|u_k\|_{\ell^2(\Lambda \setminus F')} \leq e^{-L_4}$ when $|\overline{E} - E_k| \leq e^{-L_5}$. On the other hand, $|\overline{E} - E_k(\boldsymbol{\xi}_s)| \geq e^{-L_5}$ necessarily requires 
\begin{align*}
s &\geq \frac{e^{-L_5} - e^{-L_2}}{Me^{-L_4}}\\
&\geq \frac{e^{-L_5}}{2Me^{-L_4}}
\end{align*}
In particular, as long as we take $L_5$ sufficiently large so that $\frac{e^{-L_5}}{2Me^{-L_4}} \geq 1$, we obtain the estimate:
\begin{align*}
    |\overline{E} - E_k(\boldsymbol{\xi})| &\leq |\overline{E} -E_k(\boldsymbol{\xi}')| + |E_k(\boldsymbol{\xi}') - E_k(\boldsymbol{\xi})|\\
    &\leq e^{-L_2} + 2M(L_0+1)^{3/2}e^{-2L_4}\\
    &\leq e^{-L_4}
\end{align*}
for sufficiently large $L_5$. In particular, it is straightforward (if tedious) to verify that the condition $\frac{e^{-L_5}}{2Me^{-L_4}} \geq 1$ is satisfied for
\begin{equation}
    L_5 \geq \max\{\log(2M), 2^{1.99\ve^{-1}}\}
\end{equation}
and similarly one can verify that $e^{-L_2} + 2M(L_0+1)^{3/2}e^{-2L_4} \leq e^{-L_4}$ is satisfied for
\begin{equation}
    L_5 \geq \left(\max\{ 2\log (2M), 10^3\}\right)^{(1-2\delta)^{-1}}
\end{equation}

Finally, in order to bound $|K''|$, we make use of the other part of Hypothesis \ref{qlcond}; for all $k \in K''$ we have that $\|u_k\|_{\ell^2(\Lambda)} \leq (1+L_0^{-\delta})\|u_k\|_{\ell^2(G)}$. By the ($\ell^2(\Lambda)$) orthogonality of $u|_{G}$ and $u|_{\Lambda\setminus G}$, one obtains $\|(1-P_G)u\|_{\ell^2(\Lambda)} \leq \sqrt{2L_0^{-\delta} - L_0^{-2\delta}} \leq 2L_0^{-\delta/2}$. (Here $P_G$ denotes the projection onto $G$.) Hence $\ell^2(G)$ ``almost orthogonality'' follows:
\begin{align*}
    |\langle u_{k_1}, u_{k_2}\rangle_{\ell^2(G)}-\delta_{k_1k_2}| &= |\langle P_Gu_{k_1}, P_Gu_{k_2}\rangle_{\ell^2(\Lambda)}-\delta_{k_1k_2}|\\
    &=\langle u_{k_1},u_{k_2}\rangle_{\ell^2(G)} - \delta_{k_1k_2}| + |\langle (1-P_G)u_{k_1},(1-P_G)u_{k_2}\rangle_{\ell^2(\Lambda)}|\\
    &\leq 4L_0^{-\delta/2} +4L_0^{-\delta}  \leq 5 L_0^{-\delta/2}
\end{align*}
In particular, setting $\alpha = 5$ one obtains $|K''| \leq 20L_0^{\delta}$ by \Cref{taolem}.

Hence, the lemma follows by a union bound:

\begin{align*}
    \P[\{\|(H_\Lambda - \overline{E})^{-1}\|> e^{L_1}\,|\,\sigma_u(F')] &\leq C\beta^{-5/2}L_0^{\frac{3}{2}+\delta}L_3^{-p}\\
    &\leq C\beta^{-5/2}L_0^{\frac{3}{2} +\delta -(1-\ve/2)^3p}\\
    &= C\beta^{-5/2}L_0^{\frac{3}{2} -p + \delta +\ve p/4 -\ve^2p/8}
\end{align*}
Recall that $p > 3/2$ and so in particular by taking $\ve > \delta$ sufficiently small in a way only depending on the absolute constant $p$, e.g. $\ve < \frac{1}{100}(p-\frac{3}{2})$ (and noting that clearly $p \leq 3$), one gets
\begin{align*}
    \P[\{\|(H_\Lambda - \overline{E})^{-1}\|> e^{L_1}\,|\,\sigma_u(F')] &\leq C\beta^{-5/2}L_0^{p-\frac{3}{2} + 3\ve}
\end{align*}
Finally, taking $L_5$ sufficiently large (so that $L_0^\ve > C\beta^{-5/2}$) we get
\begin{align*}
    \P[\{\|(H_\Lambda - \overline{E})^{-1}\|> e^{L_1}\,|\,\sigma_u(F')] &\leq L_0^{p-\frac{3}{2} + 4\ve}
\end{align*}
proving the main estimate in the lemma with $\ve_0 = p-\frac{3}{2}$. Moreover, the claimed dependencies on e.g. $M, \sigma^2, \ve, \delta, N$ follow from combining the various lower bounds on $L_5$ with the explicit dependencies of $\iota$ and $\beta$ on $M$ and $\sigma^2$ in \Cref{berdecomp}.
\end{proof}
\section{The multiscale analysis and proof of \Cref{mainestthm}}
By making routine adaptations to the proof of \cite[Theorem 3.10]{li2022anderson} (as was done in two dimensions, see \cite[Theorem C.1]{hurtado2026localization}), the Wegner estimate proved above combined with the initial scale estimate \Cref{initialscale} gives a multiscale analysis. We require the introduction of some definitions.

\begin{deffo}
    A scale $L \in \mathbb{N}$ is called dyadic if it is an (integral) power of two. For $L$ a dyadic scale, an $L$-cube $\Lambda$ is called dyadic if $\Lambda = \Lambda_L(x)$ for $x$ which is an integral multiple of $\frac{L}{2}$. This is sometimes called half-aligned in e.g. \cite{hurtado2026localization,ding2020localization}.
\end{deffo}
The multiscale analysis is as follows:
\begin{thm}\label{msa}
    Fix $M \geq 0$ and $\sigma^2> 0$ and a non-stationary potential satisfying the assumptions of \Cref{mainlocalthm}. All constants appearing here are allowed to depend on $M$ and $\sigma^2$. Then there is $\kappa$ such that for any $\ve' > 0$, there are
    \begin{enumerate}[label = (\Roman*)]
        \item $\ve' > \ve > \delta' > \delta > 0$
        \item $N',N \in \mathbb{N}$
        \item dyadic scales $L_k$ with $\lfloor \log_2 L_{k+1}^{1-6\ve} \rfloor = \log_2 L_k$
        \item decay rates $1 
        \geq m_k \geq L_k^{-\delta}$ for $k \geq 0$

    \end{enumerate}
    such that for any $0 \leq \overline{E} \leq \exp(-L_{N'}^{-\delta})$, there exist random sets $\mathcal{O}_k \subset \R^3$ such that the following hold:
    \begin{enumerate}[label = (\alph*)]
        \item When $k \leq N'$, $\mathcal{O}_k \cap \Z^3 = \lceil \ve^{-1} \rceil\Z^3$
        \item When $k > N'$, $\mathcal{O}_k$ is an $(N,\vec{l},(2\ve)^{-1},2\ve)$-graded random set with $\vec{l} = (L_{M+1}^{1-2\ve},L_{M+2}^{1-2\ve},\cdots,L_k^{1-2\ve})$
        \item For any $j \geq 0$ and any dyadic $2^jL_k$-cube $\Lambda$, the set $\mathcal{O}_k\cap \Lambda$ is $V_{\mathcal{O}_{k-1}\cap 3\Lambda}$-measurable
        \item We call an $L_k$-cube $\Lambda$ good if for any potential $V':\Z^3 \to [0,M]$ with $V'_{\mathcal{O}_k \cap \Lambda} = V_{\mathcal{O}_k \cap \Lambda}$ we have
        \begin{equation}
            |(H_{\Lambda}'-\overline{E})^{-1}(x,y)| \leq \exp(L_k^{1-\ve} - m_k|x-y|)
        \end{equation}
        Then any fixed $\Lambda$ is good with probability at least $1-L_k^{-\kappa}$
        \item $m_k = m_{k-1} - L_{k-1}^{-\delta'}$ for $k \geq M+1$
    \end{enumerate}
\end{thm}
This multiscale analysis in particular immediately gives us that at the scales $L_k$ appearing in the proof, one has the desired resolvent estimates. (Note that the condition of being ``good'' implies that these estimates hold in a robust way, which of course means that they hold.) Thus it remains to extract the same estimates for large scales $L$ which are not among the $L_k$; this follows essentially by the geometric resolvent inequality central to multiscale analysis. We reproduce a certain formulation of what is necessary \cite[Lemma 6.2]{ding2020localization}:
\begin{lem}[\cite{ding2020localization}]
    If
    \begin{enumerate}
        \item $\nu' > \nu > 0$ are sufficiently small
        \item $K \geq 0$ is an integer
        \item $\ell_0 \geq \cdots \geq \ell_6 \geq C_{\ve,\delta,K}$ are dyadic scales with $\ell_k^{1-\nu'} \geq \ell_{k+1}$
        \item $1 \geq m \geq 2\ell_5^{-\delta}$ is an exponential decay rate
        \item $\Upsilon$ is a $\ell_0$-cube
        \item $\Upsilon_1',\dots,\Upsilon_K'$ are defect $\ell_2$-cubes with $\|(H_{\Upsilon_k'}-E_0)^{-1}\| \leq e^{\ell_4}$
        \item for all $x \in \Upsilon$ at least one of the following holds
        \begin{enumerate}
            \item There is $\Upsilon_k'$ such that $x \in \Upsilon_k'$ and $\mathrm{dist}(x,\Upsilon\setminus \Upsilon_k') \geq \ell_2/8$
            \item There is an $L_5$-cube $\Upsilon'' \subset \Upsilon$ such that $x \in \Upsilon''$ and $\mathrm{dist}(x,\Upsilon\setminus\Upsilon'') \geq \ell_5$, and $|(H_{\Upsilon}-E_0)^{-1}(y,z)| \leq \exp(\ell_6 - m|y-z|)$ for $y,z \in \Upsilon''$
        \end{enumerate}
    \end{enumerate}
    then
    \begin{equation*}
        |(H_{\Upsilon}(x,y)-E_0)^{-1}| \leq \exp(\ell_1 - \tilde{m}|x-y|)
    \end{equation*}
    for $x,y \in \Upsilon$ where $\tilde m = m - \ell_5^{-\nu}$.
\end{lem}
\begin{remm}
    In the original statement, it was supposed that $K\geq 1$; the absence of defect cubes does not complicate the proof.
\end{remm}
Finally, we present a proof of \Cref{mainestthm} using \Cref{msa}.
\begin{proof}[Proof of \Cref{mainestthm}] As has already been discussed, because the problem is translation invariant, the desired bound has already been proven for certain scales. It remains to address general $L$, where we will incur a slight loss in our parameters. The following is quite similar to the corresponding proof of \cite[Theorem 3.1]{li2022anderson}.

    Indeed, we fix $k$ such that $L_{k+1} \leq L < L_{k+2}$, so that in particular
    \begin{equation*}
        L_k^{1+6\ve +O(\ve^2)} \leq L \leq L_k^{1+12\ve + O(\ve^2)}
    \end{equation*}
    (This requires $L \geq L_1$, which is acceptable since we have presumed $L$ sufficiently large.)
    We let
    \begin{equation*}
        \mathcal{Q} = \{\Lambda' \subset \Z^3\,:\Lambda'\text{ is a dyadic }L_k\text{-cube}\text{ such that } \Lambda'\cap \Lambda \neq \varnothing\}
    \end{equation*}
    be the collection of dyadic $L_k$ cubes non-trivially intersecting $\Lambda$. Clearly $\Lambda \subset \bigcup_{\Lambda \in \mathcal{Q}}\Lambda'$. By our supposition that these cubes are dyadic, $|\mathcal{Q}| \leq \left(\frac{2L}{L_k}+2\right)^3 \leq L_k^{49\ve + O(\ve^2)}$. Moreover, for any $a \in \Lambda$, one can find some $\Lambda' \in \mathcal{Q}$ such that $\mathrm{dist}(a,\Lambda \setminus \Lambda') \geq L_k/8$. Existence of such $\Lambda'$ is straightforward but tedious to verify by examining the cases where $a$ is ``in the bulk'' of $\Lambda$ and ``near an edge'' of $\Lambda$ separately.

    For $\Lambda' \in \mathcal{Q}$, we let $A_\Lambda'$ denote the event that
    \begin{equation*}
        |(H_{\Lambda'}-E_0)^{-1}(a,b)|\leq \exp(L^{1-\ve} - m_k|a-b|) \text{ for } a,b \in \Lambda'
    \end{equation*} By \Cref{msa}, $\P[A_\Lambda'] \geq 1 - L_k^{-\kappa}$. If we let $A = \cap_{\Lambda' \in \mathcal{Q}} A_{\Lambda'}$, then by a union bound $\P[A] \geq 1- L^{49\ve -\kappa +O(\ve^2)}$. If one has that all boxes are good (the content of the event $A$), one has resolvent estimates on the larger box $\Lambda$, by applying \cite[Lemma 6.2]{ding2020localization} with
    \begin{itemize}
        \item $\nu' = \ve/2, \nu = \delta$
        \item $K = 0$
        \item $\ell_0 = L$, $\ell_5 = L_k$, $\log_2 \ell_j = \lfloor (1-\frac34 j\ve)\log_2L\rfloor$ for $1\leq j \leq 4$ and $\log_2 \ell_6 = \lfloor (1-3\ve)\log_2 L\rfloor$
        \item $m = m_k$
        \item $\Upsilon = \Lambda$
        \item No $\Upsilon_k'$ need to be defined as $K = 0$
    \end{itemize}
    Note that if $A$ holds, then we can take for any given $a \in \Lambda$ the particular $\Lambda'$ specified above with $\mathrm{dist}(a,\Lambda'\setminus \Lambda) \geq L_k/8$ to be $\Upsilon''$, so that (if $A$ holds) we have
    \begin{equation*}
        |(H_\Lambda - E_0)^{-1}(x,y)| \leq \exp(\ell_1 - \tilde{m}|x-y|)
    \end{equation*}
    with $\ell_1 \leq  L^{1-.75 \ve}$ and $\tilde m = m_k - L_k^{-\nu}$. 

    Finally, taking $\ve$ sufficiently small such that $\kappa - 49\ve + O(\ve^2) > 0$, we take $\kappa_\ast = \frac{\kappa - 49\ve +O(\ve^2)}{1-10\ve +O(\ve^2)}$. We take $\ve_\ast = .75\ve$, and $m_\ast = m_0 - \sum_{k=0}^\infty L_0^{-\delta}$; note that $\sum_{k=0}^\infty L_0^{-\delta}$ can be made arbitrarily small by taking $L_0$ large as $L_k$ grows superexponentially, and hence we can take $m_\ast > 0$ by supposing $L_0$ sufficiently large. With the parameters so chosen, we have exactly
    \begin{equation*}
        \P\left[ |(H_\Lambda - E_0)^{-1}(x,y)| \leq \exp(L^{1-\ve_\ast} - m_\ast|x-y|)\right] \geq 1-L^{-\kappa_\ast}
    \end{equation*}

\end{proof}

Finally, we recall that \Cref{mainlocalthm} and \Cref{sdl} follow from \Cref{mainestthm} by \cite[Theorem 1]{rangamani2025dynamical}.

\bibliographystyle{amsalpha}
\bibliography{uniflyap}

\end{document}